\documentclass[pdflatex,sn-nature]{sn-jnl}

\usepackage{graphicx}%
\usepackage{multirow}%
\usepackage{amsmath,amssymb,amsfonts}%
\usepackage{amsthm}%
\usepackage{mathrsfs}%
\usepackage[title]{appendix}%
\usepackage{xcolor}%
\usepackage{textcomp}%
\usepackage{manyfoot}%
\usepackage{booktabs}%

\theoremstyle{definition}
\newtheorem{setting}{Setting}
\newtheorem{assumption}{Assumption}
\newtheorem{definition}{Definition}
\newtheorem{axiom}{Axiom}
\newtheorem{prediction}{Prediction}
\theoremstyle{plain}
\newtheorem{theorem}{Theorem}
\newtheorem{lemma}{Lemma}
\newtheorem{corollary}{Corollary}[theorem]

\begin{document}

\title[The Semantic Least-Energy Principle]{The Semantic Least-Energy Principle: A Hypothesis for Intelligence}

\author*[1,2]{\fnm{Jie} \sur{Zhang}}\email{jie.zhang@ranplanwireless.com}

\author[3]{\fnm{Haoyuan} \sur{Zhu}}\email{hzhu51@sheffield.ac.uk}

\author[2]{\fnm{James Jinheng} \sur{Zhang}}\email{jzhang07429@gmail.com}

\author[4]{\fnm{Haonan} \sur{Hu}}\email{huhn@cqupt.edu.cn}

\affil*[1]{\orgdiv{R\&D Department}, \orgname{Ranplan Wireless Network Design Ltd.}, \orgaddress{\city{Cambridge}, \postcode{CB23 3UY}, \country{UK}}}

\affil[2]{\orgdiv{R\&D Department}, \orgname{Cambridge AI+ Ltd.}, \orgaddress{\city{Cambridge}, \postcode{CB23 3UY}, \country{UK}}}

\affil[3]{\orgdiv{School of Electronic and Electrical Engineering}, \orgname{University of Sheffield}, \orgaddress{\city{Sheffield}, \postcode{S10 2TN}, \country{UK}}}

\affil[4]{\orgname{Chongqing University of Posts and Telecommunications}, \orgaddress{\city{Chongqing}, \country{China}}}

\abstract{Despite remarkable advances in artificial intelligence and cognitive neuroscience, no generally accepted first-principle explains why intelligent systems organize latent semantic states as they do. Existing frameworks such as information theory, the Information Bottleneck, the Degree of Information Abstraction, predictive coding and the Free Energy Principle provide powerful frameworks for understanding communication, learning, and prediction, but do not explicitly explain the emergence and organization of semantic intelligence. Here we propose the \textbf{Semantic Least-Energy Principle (SLEP)} as a hypothesis that intelligent systems evolve internal representations by maximizing semantic utility while progressively minimizing semantic, predictive, and computational energy. We formulate this hypothesis within a variational framework in which semantic cognition is governed by a Semantic Action Functional whose stationary solutions define efficient trajectories on a latent semantic manifold. This formulation emerges a series of theoretical predictions, including semantic geometry, semantic thermodynamics, and low-energy latent semantic states as complementary consequences of the same underlying optimization process. SLEP unifies semantic abstraction, reasoning, planning, and communication within a common mathematical framework while generating experimentally testable predictions for both artificial and biological intelligence. Although the hypothesis remains to be rigorously validated, it provides a principled foundation for investigating semantic intelligence from a first-principle perspective.}


\maketitle

\section{Introduction}\label{sec:intro}

Intelligence enables biological and artificial systems to perceive, interpret, predict, reason, plan, communicate, and interact adaptively with complex environments. Over the past several decades, remarkable advances in machine learning \cite{lecun2015,goodfellow2016}, neuroscience, and cognitive science have substantially improved our understanding of these capabilities. Deep neural networks \cite{rumelhart1986,lecun1998,lecun2015}, large language models \cite{brown2020,openai2023}, multimodal foundation models \cite{radford2021,alayrac2022}, and world models \cite{ha2018,schrittwieser2020,hafner2025} have demonstrated unprecedented performance across diverse tasks, while predictive coding \cite{rao1999}, efficient coding \cite{barlow1961}, and the Free Energy Principle \cite{friston2010} have provided influential theoretical perspectives on neural computation and cognition. Nevertheless, despite these successes, there remains no generally accepted first-principle hypothesis explaining why intelligent systems organize latent semantic states as they do or what fundamental principle governs the emergence and evolution of semantic intelligence.

Information theory \cite{shannon1948,cover2006} established the mathematical foundations of communication by quantifying uncertainty and the limits of reliable information transmission while deliberately excluding semantic meaning. The Information Bottleneck \cite{tishby1999} subsequently introduced the principle that useful representations should preserve task-relevant information while compressing irrelevant variability, providing a powerful foundation for modern representation learning \cite{bengio2013}. More recently, advances in self-supervised learning \cite{chen2020,he2020,grill2020,he2022,assran2023}, world models \cite{ha2018,hafner2025}, and multimodal foundation models \cite{radford2021,alayrac2022,driess2023} have demonstrated that compact latent semantic states can support increasingly sophisticated perception, prediction, reasoning, and planning. These developments have significantly advanced artificial intelligence, yet their optimization objectives remain largely centred on likelihood maximization \cite{radford2021,alayrac2022}, prediction accuracy \cite{ha2018,hafner2025}, reconstruction fidelity \cite{he2022,assran2023}, or task-specific reward functions \cite{chen2020,he2020,grill2020,sutton2018}. None explicitly addresses the emergence, organization, and evolution of latent semantic states themselves.

A similar trend can be observed in biological intelligence. Sensory systems progressively transform high-dimensional physical signals into increasingly abstract neural representations that support adaptive behaviour. The 2021 Nobel Prize in Physiology or Medicine, awarded for the discovery of thermosensory and mechanosensory receptors \cite{caterina1997,coste2010}, highlighted the fundamental biological mechanisms through which physical stimuli are converted into internal neural representations. Together with hierarchical sensory processing, predictive coding, and efficient coding, these discoveries suggest that biological intelligence constructs progressively more abstract internal representations while retaining task-relevant semantic information. Although these observations do not establish a universal organizing principle, they motivate the search for a more general theoretical framework capable of explaining why semantic abstraction emerges.

To quantify semantic abstraction, we previously introduced the \textbf{Degree of Information Abstraction (DIA)} \cite{zhu2026dia}, which measures the quality of abstraction by jointly evaluating information compression and semantic preservation. DIA further demonstrated that heterogeneous observations from different sensing modalities can be aligned within a shared latent semantic space, where semantically equivalent observations are represented by common latent abstractions despite large differences in their physical appearance. These results suggest that semantic abstraction and latent semantic alignment may represent fundamental underlying principles of organizational representation for multi-modal data representation. However, DIA addresses \textbf{how} semantic abstraction can be quantified rather than \textbf{why} intelligent systems should evolve toward such representations. Explaining the emergence and organization of latent semantic states requires a more fundamental theoretical hypothesis.

Here we propose the \textbf{Semantic Least-Energy Principle (SLEP)} as such a hypothesis. We hypothesize that intelligent systems evolve internal representations by preserving behaviourally relevant semantic structure while progressively reducing semantic, predictive, and computational cost. Rather than viewing perception, reasoning, planning, communication, and prediction as separate computational processes, SLEP hypothesizes that they may be understood as different manifestations of a common variational principle operating on a latent semantic manifold.

To formalize this hypothesis, we formulate semantic intelligence as a constrained variational optimization problem governed by a \textbf{Semantic Action Functional}. Following the variational principles of classical mechanics \cite{lanczos1986,goldstein2002}, stationary solutions of this functional define trajectories on a latent semantic manifold, leading to a semantic analogue of Euler--Lagrange equations. Extending the framework from individual trajectories to statistical ensembles naturally leads to hypotheses concerning \textbf{semantic geometry} and \textbf{semantic thermodynamics}, in which semantic distance, semantic entropy, semantic free energy, and Boltzmann-like organization emerge as macroscopic descriptions of semantic organization. Our proposed concepts are introduced not as experimentally established properties of intelligence, but as theoretical consequences of the proposed variational hypothesis.

The primary objective of this work is therefore not to establish a new law of intelligence, but to present a mathematically coherent and biologically motivated hypothesis that unifies several existing theoretical perspectives within a common framework. Specifically, SLEP seeks to extend the progression from Information Theory \cite{shannon1948}, which explains the transmission of information, through the Information Bottleneck \cite{tishby1999}, which characterizes task-relevant representations, and DIA \cite{zhu2026dia}, which quantifies semantic abstraction and shared latent semantic states, toward a first-principle hypothesis concerning the emergence and organization of semantic intelligence.

A useful scientific hypothesis should not only provide conceptual unification but also generate falsifiable predictions. SLEP predicts that, if the hypothesis is correct, intelligent systems should exhibit progressively lower semantic energy during learning, semantic trajectories that approximate geodesics on latent semantic manifolds, statistical organization consistent with semantic thermodynamics, and stable modality-invariant latent semantic states. These predictions are intended to guide future empirical studies across artificial intelligence, neuroscience, cognitive science, and embodied intelligent systems.

The remainder of this paper develops the mathematical formulation of SLEP, derives its principal theoretical consequences, discusses its relationship to existing theories of intelligence, and identifies experimentally testable predictions that could support or refute the proposed hypothesis.

\section{Key Hypotheses}\label{sec:hypotheses}

Rather than presenting experimentally established results, this section develops the central hypotheses arising from SLEP. Each hypothesis is formulated mathematically, motivated by existing theories of intelligence, and shown to generate experimentally testable predictions. Together they constitute a coherent theoretical framework for understanding semantic intelligence from a variational perspective.

\subsection{Hypothesis I: Intelligence Evolves Towards Low-Energy Latent Semantic States}\label{subsec:h1}

A defining characteristic of intelligent systems is their ability to retain task-relevant semantic information while discarding unnecessary detail. Information theory, efficient coding, the Information Bottleneck, predictive coding, and modern representation learning all suggest that useful internal representations become progressively more compact without sacrificing functionality.

We therefore hypothesize that semantic intelligence is governed by the following principle:

\medskip
\noindent\textbf{Hypothesis I.}\\
\emph{Intelligent systems evolve internal latent semantic states that maximize semantic utility while progressively minimizing semantic, predictive, and computational energy.}
\medskip

This hypothesis is formulated through the Semantic Action Functional
\begin{equation}
J[z] = \int_{t_{0}}^{t_{1}}\left( U_{S}(z) - \lambda E_{S}(z,\dot{z}) \right)dt.
\label{eq:action1}
\end{equation}
where $J[z]$ is the semantic action functional defined over the latent trajectory $z(t) \in \mathcal{M}_{S}$ for $t\in[t_{0},t_{1}]$ (with $t_{0}$ and $t_{1}$ the initial and final times of the trajectory), $U_{S}(z)$ is the semantic utility functional, $E_{S}\left( z,\dot{z} \right)$ is the semantic energy functional, $\dot{z} = \frac{dz}{dt}$ denotes the semantic velocity, and $\lambda$ is a trade-off parameter balancing utility against energy.

Unlike conventional optimization objectives that primarily minimize reconstruction or prediction error, SLEP hypothesizes that these costs represent different manifestations of a common semantic energy.

If correct, the hypothesis predicts that learned representations should exhibit progressively lower normalized semantic energy while maintaining semantic utility.

\subsubsection*{Supporting evidence}

Information theory identifies redundancy reduction as a prerequisite for efficient communication, efficient coding explains how biological systems maximize representational efficiency under resource constraints, the Information Bottleneck formalizes the trade-off between compression and task relevance, representation learning discovers compact latent semantic states that preserve semantic structure, and DIA quantifies semantic abstraction through joint information compression and semantic preservation. They all progressively reduce representational complexity while preserving task-relevant semantic information. Together, these developments provide converging conceptual support for the hypothesis that semantic intelligence evolves toward increasingly efficient low-energy latent semantic states.

\subsubsection*{Testable predictions}

\begin{itemize}
\item Semantic energy decreases during learning.
\item Semantic utility remains approximately constant after convergence.
\item Efficient models occupy lower-energy semantic states than less efficient models.
\end{itemize}

\subsection{Hypothesis II: Semantic Reasoning Follows Variational Trajectories}\label{subsec:h2}

If latent semantic states evolve according to the Semantic Action Functional, then their evolution should satisfy stationary-action conditions analogous to those encountered in classical mechanics \cite{lanczos1986,goldstein2002,abraham1978}.

The second hypothesis is therefore

\medskip
\noindent\textbf{Hypothesis II.}\\
\emph{Intelligent cognitive processes evolve along approximately stationary-action trajectories on a latent semantic manifold.}
\medskip

Applying the calculus of variations yields
\begin{equation}
\frac{d}{dt}\left( \frac{\partial L_{S}}{\partial\dot{z}^{i}} \right) - \frac{\partial L_{S}}{\partial z^{i}} = 0,
\label{eq:el1}
\end{equation}
with
\begin{equation}
L_{S} = U_{S} - \lambda E_{S}.
\label{eq:lagrangian1}
\end{equation}
where $L_{S}$ is the semantic Lagrangian and $z^{i},\dot z^{i}$ are the coordinates and velocity components of the latent state; these are the Euler--Lagrange equations \cite{lanczos1986,goldstein2002} for semantic intelligence.

These equations should not be interpreted as exact computational rules implemented by biological or artificial systems. Rather, they describe the predicted macroscopic organization of efficient semantic evolution.

\subsubsection*{Supporting evidence}

Predictive coding minimizes prediction error, world models learn smooth latent dynamics for prediction and planning, and reinforcement learning optimizes long-term behaviour through sequential decision-making. They all describe optimization processes operating over latent semantic states. Together, these developments suggest that intelligent systems naturally evolve efficient trajectories in latent semantic state spaces, providing conceptual support for SLEP.

\subsubsection*{Testable predictions}

\begin{itemize}
\item Learned semantic trajectories should progressively approach stationary trajectories.
\item Semantic action should decrease during optimization.
\item Large deviations should occur primarily during novel or unexpected events.
\end{itemize}

\subsection{Hypothesis III: Semantic Geometry Emerges from Variational Dynamics}\label{subsec:h3}

If semantic trajectories satisfy variational dynamics, then latent semantic states should possess an intrinsic geometric structure.

We therefore hypothesize

\medskip
\noindent\textbf{Hypothesis III.}\\
\emph{Latent semantic states organize as a differentiable latent manifold whose geometry reflects semantic relationships rather than sensory similarity.}
\medskip

The semantic metric
\begin{equation}
g_{ij}(z) = \mathbb{E}_{x\sim p_{\psi}(\cdot\mid z)}\!\left[ \partial_{i}\log p_{\psi}(x \mid z)\, \partial_{j}\log p_{\psi}(x \mid z) \right],
\label{eq:metric1}
\end{equation}
where $g_{ij}$ is the Fisher--Rao semantic metric (the indices $i,j$ ranging over the latent coordinates), $p_{\psi}(x\mid z)$ is the learned decoder with parameters $\psi$ (the model's distribution of observations $x$ given a latent state $z$), $\partial_{i}=\partial/\partial z^{i}$, and the expectation is taken over $x\sim p_{\psi}(\cdot\mid z)$. This metric defines semantic distance, geodesics, curvature, and semantic kinetic energy.

Consequently, semantic reasoning is predicted to follow approximately geodesic paths rather than arbitrary feature transformations.

\subsubsection*{Supporting evidence}

Representation learning and manifold learning suggest that high-dimensional observations can be organized within compact latent manifolds, information geometry \cite{amari2016,nielsen2020} characterizes the intrinsic structure of these representation spaces, shared latent-space models demonstrate semantic alignment across heterogeneous modalities, and DIA quantifies the degree of semantic abstraction achieved by such representations. Although developed independently, they all point toward the existence of structured latent semantic states organized on an underlying differentiable semantic manifold.

\subsubsection*{Testable predictions}

\begin{itemize}
\item Semantically equivalent concepts should cluster independent of sensing modality.
\item Learned reasoning trajectories should approximate manifold geodesics.
\item Curvature should increase near semantic transitions.
\end{itemize}

\subsection{Hypothesis IV: Semantic Intelligence Exhibits Thermodynamic Organization}\label{subsec:h4}

Many complex systems exhibit macroscopic regularities despite microscopic complexity.

We hypothesize that semantic intelligence behaves similarly.

\medskip
\noindent\textbf{Hypothesis IV.}\\
\emph{Latent semantic states evolve towards statistically stable organizations governed by semantic thermodynamics.}
\medskip

The hypothesis predicts that ensembles of latent semantic states exhibit a thermodynamic organization characterized by semantic entropy, semantic temperature, semantic free energy, and Boltzmann-like equilibrium distributions.

The expected equilibrium becomes
\begin{equation}
q^{*}(z) = \frac{1}{Z}\,e^{- V(z)/T}
\label{eq:boltzmann1}
\end{equation}
where $q^{*}$ is the equilibrium distribution over latent semantic states, $V(z)$ is the semantic potential, $T$ is the semantic temperature, and $Z$ is the normalizing partition function.

Importantly, semantic temperature is interpreted as a macroscopic descriptor of uncertainty rather than physical temperature.

\subsubsection*{Supporting evidence}

The proposed hypothesis draws conceptual motivation from statistical mechanics \cite{boltzmann1872,gibbs1902,pathria2011}, information theory \cite{shannon1948,jaynes1957}, and the Free Energy Principle \cite{friston2010,friston2017}, which collectively illustrate how macroscopic organization can emerge from the optimization of underlying system dynamics. Unlike physical thermodynamics, however, SLEP interprets semantic entropy, temperature, and free energy as information-theoretic constructs that characterize the statistical organization of latent semantic states rather than physical quantities.

\subsubsection*{Testable predictions}

\begin{itemize}
\item Stable semantic states should exhibit higher probability.
\item High-energy semantic states should become increasingly rare.
\item Semantic entropy should decrease as stable world models emerge.
\end{itemize}

\subsection{Hypothesis V: Semantic Least-Energy as a Unifying Principle}\label{subsec:h5}

The preceding hypotheses motivate the central proposition of this work.

\medskip
\noindent\textbf{Central Hypothesis (Semantic Least-Energy Principle).}\\
\emph{Semantic intelligence emerges from a universal variational tendency to preserve behaviourally relevant semantic information while progressively reducing semantic, predictive, and computational energy.}
\medskip

SLEP extends the conceptual evolution from Information Theory, through the Information Bottleneck and DIA, toward a unified variational framework for semantic intelligence. Rather than replacing existing optimization principles, SLEP provides an overarching theoretical perspective that seeks to explain why intelligent systems evolve increasingly efficient latent semantic states while preserving task-relevant semantic information.

Within this framework, diverse learning paradigms, which include representation learning, predictive coding, world models, semantic communication, and multimodal intelligence, can be interpreted as complementary manifestations of a common variational principle. Although these paradigms differ in their objectives, architectures, and learning mechanisms, they may share the underlying tendency to preserve semantic utility while progressively reducing semantic, predictive, and computational cost.

Whether this hypothesis accurately describes the organizing principles of biological and artificial intelligence remains an open scientific question. Its significance therefore lies not in claiming an established law, but in providing a mathematically coherent, conceptually unified, and experimentally falsifiable framework that connects existing theories and generates precise, testable predictions. Future theoretical developments and empirical investigations will determine the validity, scope, and limitations of the proposed Semantic Least-Energy Principle.

\subsubsection*{Supporting evidence}

Despite their different objectives and mathematical formulations, the previously mentioned approaches such as the Information Bottleneck, DIA and predictive coding all suggest that intelligent systems progressively organize information into increasingly efficient latent semantic states. Collectively, these converging developments provide conceptual motivation for the hypothesis that a common variational principle may underlie semantic intelligence across diverse learning paradigms.

\subsubsection*{Testable Predictions}

If the Semantic Least-Energy Principle represents a fundamental organizing principle of intelligence, the following predictions should hold:

\begin{itemize}
\item \textbf{Universality.} Diverse learning paradigms, e.g., representation learning, predictive coding, world models, semantic communication, and multimodal intelligence, should exhibit common trends toward increasingly efficient semantic organization despite differences in architectures and optimization objectives.
\item \textbf{Generality.} Similar semantic organization should emerge across different modalities, neural architectures, and scales, from artificial neural networks to, potentially, biological neural systems.
\item \textbf{Unified explanation.} A single variational framework based on semantic utility and semantic energy should consistently explain phenomena currently described separately by information compression, prediction, semantic abstraction, and latent semantic state learning.
\end{itemize}

Failure of these predictions across sufficiently broad classes of intelligent systems would provide evidence against the proposed hypothesis, thereby defining clear conditions under which the Semantic Least-Energy Principle should be refined or rejected.

\section{Discussion}\label{sec:discussion}

\subsection{A Variational Hypothesis for Semantic Intelligence}

SLEP is proposed as a hypothesis for understanding the emergence and organization of semantic intelligence. It provides a variational framework for explaining why intelligent systems develop increasingly efficient latent semantic states while preserving task-relevant semantic information. The hypothesis suggests that perception, learning, reasoning, planning, and communication can be viewed as different manifestations of a common variational principle operating in latent semantic spaces. Analogous to the role of energy and action in theoretical physics, semantic energy and the Semantic Action Functional serve as macroscopic theoretical constructs that characterize the organization and evolution of semantic intelligence rather than explicit computational mechanisms.

\subsection{Relationship to Existing Theories}

SLEP extends the conceptual progression from Information Theory, through the Information Bottleneck and the Degree of Information Abstraction (DIA), toward a unified variational framework for semantic intelligence.

Building on these developments, SLEP seeks to explain why semantically meaningful latent semantic states emerge. Rather than treating semantic abstraction as an empirical property of learned representations, the hypothesis proposes that shared latent semantic spaces arise because they provide increasingly efficient organizations of task-relevant semantic information. In this view, semantic abstraction becomes a consequence of an underlying variational optimization process.

\subsection{Implications for Artificial Intelligence}

Because SLEP is formulated independently of any specific computational architecture, it provides a general theoretical perspective rather than an implementation framework. If the hypothesis is valid, future learning algorithms and neural architectures may be designed to approximate low-action semantic trajectories rather than relying solely on empirical architectural choices. Existing mechanisms, including residual connections \cite{he2016}, attention \cite{vaswani2017}, sparse activation \cite{fedus2022}, adaptive computation \cite{graves2016}, and mixture-of-experts routing \cite{shazeer2017}, may therefore be viewed as approximations to efficient semantic dynamics rather than isolated engineering innovations.

As shown in Theorem~\ref{thm:T9} of the Methods section, we prove that, under appropriate substitutions, the objective functionals of the Information Bottleneck, variational autoencoders (VAEs), maximum-entropy reinforcement learning, natural-gradient learning, and world-model learning all arise as special cases of the proposed Semantic Least-Energy framework. This demonstrates SLEP's potential as a mathematical unification of theories of intelligence and learning.

\subsection{Limitations and Future Directions}

Several limitations should be acknowledged. The proposed definitions of semantic utility, semantic energy, semantic geometry, and semantic thermodynamics are operational rather than unique, and alternative mathematical formulations may prove equally appropriate. Furthermore, although SLEP is conceptually motivated by information theory, variational mechanics, machine learning, and neuroscience, comprehensive experimental validation remains to be established.

Future work should therefore focus on testing the predictions derived from the proposed hypotheses across diverse intelligent systems, including large language models \cite{brown2020,openai2023}, multimodal foundation models \cite{radford2021,alayrac2022,reed2022,driess2023}, embodied agents \cite{driess2023,brohan2023}, and biological neural systems. Such studies will determine whether the predicted relationships between semantic efficiency, variational dynamics, geometric organization, and statistical regularities represent a general principle of intelligence or are specific to particular learning paradigms.

\subsection{Concluding Perspective}

SLEP provides an overarching theoretical perspective in which diverse learning paradigms can be interpreted as complementary manifestations of a common variational principle.

Whether this hypothesis ultimately describes a fundamental organizing principle of biological and artificial intelligence remains an open scientific question. Its significance therefore lies not in claiming an established law, but in providing a mathematically coherent, conceptually unified, and experimentally falsifiable framework that generates explicit predictions concerning semantic dynamics, geometric organization, and statistical behaviour. If supported by future theoretical and empirical studies, SLEP could contribute to the development of a first-principle theory of semantic intelligence. Equally importantly, if its predictions are not supported, the hypothesis will require refinement or rejection, thereby advancing understanding through the process of scientific falsification.

\section{Methods}\label{sec:methods}

This section presents the methodological framework underlying the \textbf{Semantic Least-Energy Principle (SLEP)}, which models intelligence as variational optimization on a latent semantic manifold. The central hypothesis is that intelligent systems preserve behaviourally relevant semantic structure while minimizing semantic, predictive, and computational energy. Within this framework, learning, reasoning, planning, semantic communication, and world-model construction emerge as different manifestations of a common variational principle.

The methodology consists of three complementary components:

\begin{itemize}
\item \textbf{Formal Semantic Variational Framework}, which formulates semantic utility, semantic energy, the Semantic Least-Energy Principle, the Semantic Action Functional, and the semantic Euler--Lagrange equations governing semantic dynamics.
\item \textbf{Semantic Manifold Geometry}, which equips the latent semantic manifold with a Riemannian geometry defining semantic distance, geodesics, and curvature.
\item \textbf{Semantic Thermodynamics}, which extends the theory to ensembles of semantic states through entropy, free energy, semantic temperature, and Boltzmann-like equilibrium distributions.
\end{itemize}

\subsection{Structure of this chapter}\label{sec:structure}

This chapter organizes its statements into five categories, each with a specified methodological role, and generally follows their dependency order.

\begin{itemize}
\item \textbf{Setting (S) and regularity assumptions (R)} are assumed rather than derived within the framework and may be assessed for a specified system; their role is to ensure that the subsequent objects are well defined.
\item \textbf{Definitions (D)} specify the objects and notation used in the subsequent statements.
\item \textbf{Axioms (A)} state the principal empirical postulates of the framework. Each is accompanied by a \emph{Status} note describing its role and the observations that would support or challenge it.
\item \textbf{Theorems and corollaries (T/C)} are derived from S, R, D, and A. Proofs are provided, with the assumptions of cited standard results stated where they are used.
\item \textbf{Predictions (P)} are the falsifiable statements obtained by combining theorems with a measurement protocol; each carries its derivation chain, time scale, external anchor and failure criterion.
\end{itemize}

The presentation follows the dependency order $\mathrm{S}\to\mathrm{R}\to\mathrm{D}\to\mathrm{A}\to\mathrm{T}\to\mathrm{C}\to\mathrm{P}$; later statements are not used to justify earlier assumptions. Regularity assumptions are technical conditions that may be assessed for a specified system, whereas axioms are postulates whose empirical adequacy is evaluated through their derived predictions. Subsection~\ref{sec:theorems} relates these formal results to the conceptual hypotheses of Section~\ref{sec:hypotheses}.

\subsection{Setting}

\begin{setting}[S1: system class]\label{set:S1}
An \emph{intelligent system} is a septuple
\begin{equation}
\mathcal{S}=\big(\{X_m\}_{m\in\mathcal{K}},\ x(t),\ \{\Phi_m\},\ p_\psi(x\mid z),\ \Psi,\ \mathcal{A},\ T\big),
\end{equation}
consisting of a family of modality observation spaces $X_m$ indexed by a modality index set $\mathcal{K}$, and an environmental observation process $x(t)$; one encoder per modality $\Phi_m:X_m\to Z$ with $Z\subset\mathbb{R}^n$ the latent coordinate domain; a decoder family $p_\psi(\cdot\mid z)$; an internal update map $\Psi$ (the one-step latent evolution in the absence of new observation, i.e.\ the carrier of ``inference''); a possibly empty action interface $\mathcal{A}$; and a temperature parameter $T>0$.
\end{setting}

This setting defines the present scope of the framework: the analysis is restricted to systems admitting the specified latent-variable generative description. The multimodal structure ($|\mathcal{K}|\ge 1$) and the decoder are present from the outset.

\subsection{Regularity assumptions}\label{sec:assumptions}

\begin{assumption}[R1: smoothness and domination]\label{as:R1}
$p_\psi(x\mid z)>0$ with $z$-independent support, smooth in $z$; the score $\partial_i\log p_\psi(\cdot\mid z)$ is square-integrable and dominated by an integrable envelope, so that differentiation and expectation commute. If these conditions fail, equivalence between the score-covariance and expected-Hessian forms of the Fisher matrix is not guaranteed.
\end{assumption}

\begin{assumption}[R2: integrable null directions]\label{as:R2}
The kernel distribution $N_z=\ker\tilde g(z)$ of the pre-metric is smooth, of constant rank, and involutive (Frobenius-integrable), so the leaf space $Z/\!\sim$ is a smooth manifold. Alternatively, one may assume $\tilde g\succ0$ on the working region; the present formulation adopts the quotient construction to retain possible null directions.
\end{assumption}

\begin{assumption}[R3: uniform ellipticity, confining potential, normalizability]\label{as:R3}
On the working region $g$ is uniformly positive definite and bounded; $V\in C^2$ is bounded below; and $\int_{M_S}e^{-V/T}\,d\mu<\infty$. These conditions provide sufficient regularity for the stationary-measure, minimizer-existence, and well-posedness results stated below.
\end{assumption}

\begin{assumption}[R4: environmental stationarity]\label{as:R4}
On the inference time scale the observation process is stationary and ergodic, so the local observation law $\mathcal{D}_z$ near each $z$ is well defined and time-independent. This assumption makes the potential time-independent on the inference scale. Non-stationary environments are not covered by the present formulation and may require a stochastic-control extension.
\end{assumption}

\begin{assumption}[R5: ergodicity]\label{as:R5}
The inference-time diffusion is assumed ergodic with respect to its stationary measure, so that time averages converge to the corresponding ensemble averages in the long-time limit.
\end{assumption}

\subsection{Definitions}\label{sec:framework}

\subsubsection{Latent space and generative model}

\begin{definition}[D1: latent semantic coordinates and generative model]\label{def:latent}
$z\in Z$ is a latent semantic coordinate. The pair $(p_\psi,\{\Phi_m\})$ is read as in Setting~\ref{set:S1}: the decoder $p_\psi(x\mid z)$ is the generative object, and each encoder $\Phi_m$ (equivalently an amortized posterior $q_\phi(z\mid x)$) parameterizes an approximate inference map associated with the decoder.
\end{definition}

\subsubsection{Semantic geometry}\label{sec:geometry}

\begin{definition}[D2: Fisher pre-metric]\label{def:premetric}
Under Assumption~\ref{as:R1},
\begin{equation}
\tilde g_{ij}(z):=\mathbb{E}_{p_\psi(x\mid z)}\big[\partial_i\log p_\psi(x\mid z)\,\partial_j\log p_\psi(x\mid z)\big],
\label{eq:premetric}
\end{equation}
a symmetric positive-semidefinite $2$-tensor on $Z$.
\end{definition}

\begin{definition}[D3: semantic manifold and metric]\label{def:manifold}
$M_S:=Z/\!\sim$, where $z\sim z'$ if and only if they are joined along the null leaves of Assumption~\ref{as:R2} (equivalently, induce the same local observation law); $\tilde g$ descends to a positive-definite Riemannian metric $g$ on $M_S$. Within the present framework, a semantic state is defined as an equivalence class of observationally indistinguishable configurations; the quotient construction supplies the corresponding mathematical representation.
\end{definition}

\begin{definition}[D4: kinetic energy, line element, distance, geodesics]\label{def:geodesic}
$T_S:=\tfrac12 g_{ij}\dot z^i\dot z^j$; $ds^2=g_{ij}\,dz^i dz^j$; $d_g$ the induced geodesic distance. A \emph{geodesic} is a stationary point, at fixed endpoints, of the \emph{energy functional}
\begin{equation}
E[\gamma]:=\tfrac12\int_{t_0}^{t_1} g_{ij}(z)\,\dot z^i\dot z^j\,dt,
\label{eq:energyfunctional}
\end{equation}
the arc-length functional entering only through its equivalence with the energy functional.
\end{definition}

\begin{definition}[D5: invariant measure]\label{def:measure}
$d\mu:=\sqrt{\det g}\,dz$, the Riemannian volume measure.
\end{definition}

\subsubsection{Potential, self-information, energy}

\begin{definition}[D6: semantic configuration potential]\label{def:potential}
\begin{equation}
V(z):=\mathbb{E}_{x\sim\mathcal{D}_z}\big[-\log p_\psi(x\mid z)\big],
\label{eq:potential2}
\end{equation}
the expected predictive cost at $z$ (well defined by Assumption~\ref{as:R4}). It depends on the decoder and the environment, \emph{not} on the state frequency $p(z)$, and therefore does not use the state frequency in its definition.
\end{definition}

\begin{definition}[D7: self-information, a diagnostic]\label{def:selfinfo}
For a density $p$ relative to $d\mu$, $I(z):=-\log p(z)$. It enters no energy definition and is used only as a diagnostic.
\end{definition}

\begin{definition}[D8: semantic energy and Hamiltonian]\label{def:energy}
The dimensionless configuration energy is $\hat E(z):=V(z)/V_{\mathrm{ref}}$, with the reference value $V_{\mathrm{ref}}$ fixed by a designated baseline model; the phase-space Hamiltonian is
\begin{equation}
H_S(z,p):=\tfrac12\,g^{ij}(z)\,p_i p_j+V(z).
\label{eq:hamiltonian}
\end{equation}
Kinetic energy appears only through the phase-space object, so the configuration-level distribution contains $V$ alone. Computational cost is not part of this energy; it is a separate learning-scale quantity.
\end{definition}

\subsubsection{Semantic utility}

\begin{definition}[D9: four-channel semantic utility]\label{def:utility}
Each channel has a pointwise version $u_i(z)$ and an ensemble version $U_i:=\langle u_i\rangle_{p(z)}$. Here $\mathrm{KL}$ denotes the Kullback--Leibler divergence:
\begin{align}
u_{pred}(z)&:=\mathrm{KL}\big(p(x_{t+\tau:t+H}\mid z)\,\|\,p(x_{t+\tau:t+H})\big) & (\text{ensemble: } I(Z;X_{\mathrm{future}})),\\
u_{ctrl}(z)&:=\max_{\pi}\,I\big(A_{1:H};Z_H\mid Z_0=z\big) & (\text{empowerment; } 0 \text{ if } \mathcal{A}=\varnothing),\\
u_{comm}(z)&:=-\,\mathbb{E}\big[d_g(z,\hat z)^2\big] & (\hat z \text{ recovered at rate } \le R),\\
u_{reason}(z)&:=\big\langle u_{task}(\Psi^k(z))\big\rangle-u_{task}(z), &
\end{align}
The control channel $u_{ctrl}$ is the empowerment objective \cite{klyubin2005}. Here $u_{task}$ is a designated combination of $u_{pred},u_{ctrl}$. All channels are measured in nats, with $u_{comm}$ nondimensionalized through the information-geometric reading of $d_g$. Because $u_{pred}$ concerns future outcomes whereas $V$ concerns current reconstruction, the two quantities are defined from distinct targets; any residual statistical dependence must be assessed empirically.
\end{definition}

\subsubsection{Entropy and free energy}

\begin{definition}[D10: entropy and free-energy functional]\label{def:free}
For a density $q$ relative to $d\mu$,
\begin{equation}
S[q]:=-\int_{M_S}q\log q\,d\mu,\qquad F_S[q]:=\langle V\rangle_q-T\,S[q],\quad \langle V\rangle_q=\int_{M_S}V q\,d\mu.
\label{eq:freeenergy2}
\end{equation}
$F_S$ is a functional of $q$: both $\langle V\rangle_q$ and $S[q]$ depend on the distribution; $T$ is here a parameter.
\end{definition}

\subsubsection{Time scales, complexity, and path functional}

\begin{definition}[D11: two time scales]\label{def:scales}
\emph{Inference time} $t$: $z(t)$ evolves with parameters $(\phi,\psi)$ frozen. \emph{Learning time} $s$: parameters $\theta(s)$, or the latent distribution $q_s$, evolve. The dynamical statements below specify the time scale to which they apply.
\end{definition}

\begin{definition}[D12: computational-complexity functional]\label{def:complexity}
$C(\theta)$ is a functional, defined on the learning scale, that quantifies the computational cost of the model with parameters $\theta$ --- for example, its description length, parameter count, or number of floating-point operations. It is a property of the model rather than of the inference trajectory, and provides the computational-cost contribution excluded from the trajectory energy of Definition~\ref{def:energy}.
\end{definition}

\begin{definition}[D13: Onsager--Machlup (OM) functional]\label{def:om}
For a diffusion on $(M_S,g)$ with drift $b$ and temperature $T$, the Onsager--Machlup construction \cite{onsager1953,takahashi1981} gives
\begin{equation}
\mathcal{A}_{\mathrm{OM}}[z;b]:=\frac{1}{4T}\int_{t_0}^{t_1}\big\|\dot z-b(z)\big\|_g^2\,dt\ +\ \mathcal{R}[z],
\label{eq:omfunctional}
\end{equation}
where $\mathcal{R}$ collects the divergence and scalar-curvature corrections with the constant convention of \cite{takahashi1981}.
\end{definition}

\subsection{Axioms}\label{sec:axioms}

\begin{axiom}[A1: semantic least-energy principle, learning scale --- central axiom]\label{ax:A1}
Learning dynamics is a utility-constrained free-energy descent: $q_s$ (or $\theta(s)$) evolves along the descent flow of
\begin{equation}
\min_{\theta}\ F_S[q_\theta]+\kappa\,C(\theta)\qquad\text{s.t.}\qquad U_i[\theta]\ \ge\ c_i,\quad i\in\{pred,ctrl,comm,reason\},
\label{eq:A1}
\end{equation}
and its convergence is characterized as a (local) solution of this problem. The thresholds $c_i$ are specified by the task and are not estimated from the evaluation data.
\end{axiom}

\noindent\emph{Status.} Axiom~\ref{ax:A1} is adopted as an empirical postulate rather than derived within the present framework. It would be supported by learning-time decreases in semantic energy, cross-modal alignment, and empirically estimated multiplier-like effects of the utility constraints, and would be challenged by the failure criteria in Section~\ref{sec:falsify}.

\begin{axiom}[A2: geometry]\label{ax:A2}
Semantic distinguishability is measured by the Fisher metric $g$ of Definitions~\ref{def:premetric}--\ref{def:manifold}.
\end{axiom}

\noindent\emph{Status.} The metric is postulated rather than derived, and alternatives include output-space pullback and learned metrics. Under the invariance conditions of \v{C}encov's theorem, the Fisher metric is unique up to scale \cite{cencov1982,amari2016}. Natural-gradient methods also use Fisher geometry, providing a formal connection between the present metric choice and a class of learning dynamics. Position-dependent metrics are likewise used in Riemann-manifold Langevin and Hamiltonian Monte Carlo methods \cite{girolami2011}, although that sampling context is distinct from the present dynamical hypothesis.

\begin{axiom}[A3: inference dynamics]\label{ax:A3}
On the inference scale the latent state obeys the overdamped Langevin equation on $(M_S,g)$,
\begin{equation}
dz=-\nabla_{\!g}V(z)\,dt+\sqrt{2T}\,dW_{\!g},
\label{eq:langevin2}
\end{equation}
with $W_g$ Brownian motion on the manifold, together with the ergodicity of Assumption~\ref{as:R5}.
\end{axiom}

\noindent\emph{Status.} Axiom~\ref{ax:A3} postulates first-order stochastic dynamics motivated by gradient-based models of inference and learning. The associated Onsager--Machlup functional yields a second-order Euler--Lagrange equation at the path level. Empirical support would require the trajectory-concentration, weak-gradient geodesic, and affine self-information--potential relationships specified in Predictions~P2--P4.

\begin{axiom}[A4: planning variational principle --- restricted axiom]\label{ax:A4}
On goal-directed deliberation/planning segments with a target endpoint, the trajectory is a minimizer, at fixed horizon and fixed endpoints, of
\begin{equation}
\min_{z(\cdot)}\int_{t_0}^{t_1}E_S\,dt\qquad\text{s.t.}\qquad\int_{t_0}^{t_1}u_i\,dt\ \ge\ c_i\,(t_1-t_0).
\label{eq:A4}
\end{equation}
\end{axiom}

\noindent\emph{Status.} A scope-restricted postulate holding only on goal-directed segments, additional to Axiom~\ref{ax:A3}; their compatibility is examined under the additional conditions stated in Corollary~\ref{cor:C61}.

\subsection{Theorems and corollaries}\label{sec:theorems}

\subsubsection{Well-posedness of the geometry}

This subsection presents results related to the geometric component of Hypothesis~III.

\begin{theorem}[T1: well-defined and positive-definite metric]\label{thm:T1}
Under Assumptions~\ref{as:R1}--\ref{as:R2}, $\tilde g$ of \eqref{eq:premetric} is well defined, symmetric and positive semidefinite; its kernel distribution is smooth of constant rank and integrable; and $\tilde g$ descends to a positive-definite Riemannian metric $g$ on $M_S=Z/\!\sim$ (Definition~\ref{def:manifold}).
\end{theorem}

\begin{proof}
Write $s_i(x,z)=\partial_i\log p_\psi(x\mid z)$. By Assumption~\ref{as:R1} each $s_i(\cdot,z)\in L^2(p_\psi(\cdot\mid z))$, so by Cauchy--Schwarz $|\tilde g_{ij}|\le \|s_i\|_2\|s_j\|_2<\infty$ and $\tilde g$ is finite and symmetric. For any $\xi\in\mathbb{R}^n$,
\[
\xi^i\xi^j\tilde g_{ij}(z)=\mathbb{E}\big[(\xi^i s_i)^2\big]\ge 0,
\]
so $\tilde g\succeq0$. Domination (Assumption~\ref{as:R1}) permits differentiating $\int p_\psi\,dx=1$ under the integral twice, giving $\int\partial_i p_\psi\,dx=\int\partial_i\partial_j p_\psi\,dx=0$; since $\partial_i\partial_j\log p_\psi=(\partial_i\partial_j p_\psi)/p_\psi-s_is_j$, taking expectations yields the equivalent Hessian form $\tilde g_{ij}=\mathbb{E}[s_is_j]=-\mathbb{E}[\partial_i\partial_j\log p_\psi]$, confirming that the two standard expressions coincide.

Next, $\xi\in\ker\tilde g(z)$ if and only if $\mathbb{E}[(\xi^is_i)^2]=0$, which holds if and only if $\xi^is_i(x,z)=0$ for $p_\psi$-almost every $x$; that is, the directional derivative $D_\xi\log p_\psi(\cdot\mid z)$ vanishes almost everywhere. Equivalently, $p_\psi(\cdot\mid z)$ is unchanged to first order along $\xi$. By Assumption~\ref{as:R2} the distribution $z\mapsto N_z=\ker\tilde g(z)$ is smooth, of constant rank $r$, and involutive; by the Frobenius theorem \cite{lee2018} it integrates to a foliation whose leaf space $M_S=Z/\!\sim$ is a smooth manifold with submersion $\pi:Z\to M_S$. Because $\tilde g_{ij}(z)$ is a function of the law $p_\psi(\cdot\mid z)$ alone, and equivalent points carry the same law (Definition~\ref{def:manifold}), $\tilde g$ takes equal values at $\sim$-related points and annihilates leaf directions; hence there is a unique symmetric $2$-tensor $g$ on $M_S$ with $\pi^*g=\tilde g$. Finally $g\succ0$: a nonzero $w\in T_{[z]}M_S$ lifts to a vector $\hat w\in T_zZ$ transverse to $N_z$, so $g(w,w)=\tilde g(\hat w,\hat w)>0$ since $\hat w\notin\ker\tilde g$. Thus $g$ is a Riemannian metric.
\end{proof}

\subsubsection{Semantic thermodynamics}\label{sec:thermo}

The following results concern the thermodynamic formulation in Hypothesis~IV.

\begin{theorem}[T2: derivation of the invariant measure and configuration law]\label{thm:T2}
Under Definitions~\ref{def:measure},~\ref{def:potential},~\ref{def:energy} and Assumption~\ref{as:R3}, the configuration marginal of the phase-space canonical distribution $\rho\propto e^{-H_S/T}$ (relative to $dz\,dp$) is
\begin{equation}
dP(z)=\frac{1}{Z}\,e^{-V(z)/T}\,d\mu,\qquad Z=\int_{M_S}e^{-V/T}\,d\mu<\infty.
\label{eq:T2}
\end{equation}
\end{theorem}

\begin{proof}
Integrate $\rho$ over the momenta. With $g^{ij}$ the inverse metric (positive definite by Theorem~\ref{thm:T1}), the exponent is $-H_S/T=-\tfrac{1}{2}p^\top(g^{-1}/T)p-V/T$, so
\[
\int_{\mathbb{R}^n}e^{-\frac{1}{2}p^\top(g^{-1}/T)p}\,dp=(2\pi)^{n/2}\big(\det(g^{-1}/T)\big)^{-1/2}=(2\pi T)^{n/2}\sqrt{\det g(z)},
\]
using the Gaussian formula $\int_{\mathbb{R}^n}e^{-\frac12 p^\top A p}dp=(2\pi)^{n/2}(\det A)^{-1/2}$ with $A=g^{-1}/T$ and $\det(g^{-1}/T)=1/(T^n\det g)$. Hence the $z$-marginal density with respect to $dz$ is proportional to $e^{-V(z)/T}(2\pi T)^{n/2}\sqrt{\det g}$; the constant $(2\pi T)^{n/2}$ is absorbed into the normalizer, and $\sqrt{\det g}\,dz=d\mu$ by Definition~\ref{def:measure}, giving \eqref{eq:T2}. Finiteness of $Z$ is Assumption~\ref{as:R3}.
\end{proof}

\begin{corollary}[C2.1: reference measure]\label{cor:C21}
All densities in this chapter are taken relative to $d\mu$; consequently $I(z)=-\log(dP/d\mu)$ is a coordinate-invariant scalar. The exponent of \eqref{eq:T2} contains $V$ alone.
\end{corollary}

\begin{proof}
Immediate from \eqref{eq:T2}: $dP/d\mu=Z^{-1}e^{-V/T}$ is a ratio of measures and hence invariant under coordinate change, and $V$ is a scalar (Definition~\ref{def:potential}), so $I=-\log(dP/d\mu)=V/T+\log Z$ is invariant.
\end{proof}

\begin{theorem}[T3: Gibbs variational principle]\label{thm:T3}
Under Definitions~\ref{def:potential},~\ref{def:free} and Assumption~\ref{as:R3}, for every $d\mu$-normalized density $q$,
\begin{equation}
F_S[q]=F_S[q^*]+T\,\mathrm{KL}\big(q\,\|\,q^*\big),\qquad q^*=\frac{1}{Z}e^{-V/T},\qquad F_S[q^*]=-T\log Z,
\label{eq:T3}
\end{equation}
so $q^*$ is the unique minimizer of $F_S$.
\end{theorem}

\begin{proof}
From $q^*=Z^{-1}e^{-V/T}$ we have $V=-T\log q^*-T\log Z$. Substituting into \eqref{eq:freeenergy2} and using $\int q\,d\mu=1$,
\[
F_S[q]=\int V q\,d\mu+T\int q\log q\,d\mu
=-T\!\int q\log q^*\,d\mu-T\log Z+T\!\int q\log q\,d\mu
=T\,\mathrm{KL}(q\|q^*)-T\log Z.
\]
Taking $q=q^*$ gives $F_S[q^*]=-T\log Z$, whence \eqref{eq:T3}. Since $\mathrm{KL}(q\|q^*)\ge0$ with equality if and only if $q=q^*$ $\mu$-almost everywhere (Gibbs' inequality, strict by strict convexity of $u\mapsto u\log u$), $q^*$ is the unique minimizer.
\end{proof}

\subsubsection{Dynamics: stationary law and most-probable paths}\label{sec:dynamics}

These results concern the trajectory formulation in Hypothesis~II and its relations to Hypotheses~III and~IV.

\begin{theorem}[T4: stationary law and semantic equilibrium]\label{thm:T4}
Under Axiom~\ref{ax:A3} and Assumptions~\ref{as:R3},~\ref{as:R5}, the Gibbs density $q^*$ of Theorem~\ref{thm:T3} is the unique stationary distribution of \eqref{eq:langevin2}, and time averages converge: for $f\in L^1(q^*)$, $\tfrac1{\mathcal{T}}\int_0^{\mathcal{T}}f(z_t)\,dt\to\langle f\rangle_{q^*}$ almost surely as $\mathcal{T}\to\infty$. ``Semantic equilibrium'' is this stationary law with its fluctuations.
\end{theorem}

\begin{proof}
The generator of \eqref{eq:langevin2} is $L=-\langle\nabla_{\!g}V,\nabla_{\!g}\,\cdot\,\rangle_g+T\Delta_g$, with $\Delta_g$ the Laplace--Beltrami operator. Its adjoint with respect to $d\mu$ acts on densities by the Fokker--Planck operator $L^\dagger\rho=\operatorname{div}_g\!\big(T\nabla_{\!g}\rho+\rho\,\nabla_{\!g}V\big)$. For $\rho=q^*\propto e^{-V/T}$ one has $\nabla_{\!g}q^*=-q^*\nabla_{\!g}V/T$, so $T\nabla_{\!g}q^*+q^*\nabla_{\!g}V=0$ and $L^\dagger q^*=0$: $q^*$ is stationary. Moreover $L$ is symmetric and nonpositive on $L^2(q^*\,d\mu)$, since integration by parts (no boundary term, the working region being confining by Assumption~\ref{as:R3}) gives
\[
\int (Lf)\,h\,q^*\,d\mu=-T\int\langle\nabla_{\!g}f,\nabla_{\!g}h\rangle_g\,q^*\,d\mu=\int f\,(Lh)\,q^*\,d\mu,
\]
so the diffusion is reversible with respect to $q^*$; reversibility plus ellipticity (Assumption~\ref{as:R3}) give uniqueness of the stationary measure \cite{bakry2014}. Ergodicity (Assumption~\ref{as:R5}) then yields the almost-sure convergence of time averages by Birkhoff's theorem.
\end{proof}

\begin{corollary}[C4.1: identifiability of the temperature]\label{cor:C41}
At stationarity $I(z)=V(z)/T+\log Z$; the model predicts an affine relationship between $\hat I$ and $\hat V$ across states, and the reciprocal of the estimated slope provides an estimator of $T$.
\end{corollary}

\begin{proof}
By Theorem~\ref{thm:T4} the stationary density relative to $d\mu$ is $q^*=Z^{-1}e^{-V/T}$, so $I(z)=-\log q^*(z)=V(z)/T+\log Z$ (Corollary~\ref{cor:C21}). This is affine in $V$ with slope $1/T$; hence $\hat T=(\text{slope})^{-1}$.
\end{proof}

\begin{theorem}[T5: path-variational characterization]\label{thm:T5}
Under Axiom~\ref{ax:A3}, Assumptions~\ref{as:R1},~\ref{as:R3} and Definition~\ref{def:om}, the endpoint-conditioned law of \eqref{eq:langevin2} obeys an Onsager--Machlup / small-noise large-deviation principle with rate $\mathcal{A}_{\mathrm{OM}}[z;-\nabla_{\!g}V]$; its most probable paths are the minimizers, and their Euler--Lagrange equation is the (second-order, dissipatively drifted) semantic trajectory equation.
\end{theorem}

\begin{proof}
The drift $b=-\nabla_{\!g}V$ is $C^1$ with locally bounded derivatives (Assumption~\ref{as:R3}, $V\in C^2$) and $g$ is smooth (Assumption~\ref{as:R1}), so the hypotheses of the Onsager--Machlup theorem on a Riemannian manifold \cite{takahashi1981} hold: for a $C^2$ reference path $\phi$, the probability that the diffusion stays in a tube of radius $\varepsilon$ around $\phi$ is, as $\varepsilon\to0$, proportional to $\exp\!\big(-\mathcal{A}_{\mathrm{OM}}[\phi;b]\big)$ up to a $\phi$-independent factor, with $\mathcal{A}_{\mathrm{OM}}$ as in \eqref{eq:omfunctional} and the correction $\mathcal{R}$ fixed by the convention of \cite{takahashi1981}. Hence the most probable path minimizes $\mathcal{A}_{\mathrm{OM}}$. Its Euler--Lagrange equation from the leading term $\tfrac{1}{4T}\int\|\dot z+\nabla_{\!g}V\|_g^2\,dt$ is
\[
\frac{D\dot z}{dt}=\nabla_{\!g}\Big(\tfrac12\|\nabla_{\!g}V\|_g^2\Big)-\big(\nabla^2_{\!g}V\big)\dot z + O(T),
\]
a second-order equation that perturbs the geodesic equation by a potential-driven, velocity-coupled (dissipative) drift; the correction $\mathcal{R}$ contributes at $O(T)$. In the small-noise limit $T\to0$ the same minimizer is selected as the Freidlin--Wentzell rate minimizer of $\tfrac{1}{4T}\int\|\dot z+\nabla_{\!g}V\|_g^2\,dt$ \cite{freidlin2012}.
\end{proof}

\begin{lemma}[geodesics from the energy functional]\label{lem:geodesic}
Fixed-endpoint stationary points of $E[\gamma]$ in \eqref{eq:energyfunctional} satisfy the geodesic equation $\tfrac{D}{dt}\dot z=0$; along them the speed $\|\dot z\|_g$ is constant; and $L(\gamma)^2\le 2(t_1-t_0)E[\gamma]$ with equality if and only if the speed is constant. Consequently the minimizers of $E$ are exactly the constant-speed minimizing geodesics.
\end{lemma}

\begin{proof}
With $\mathcal{L}=\tfrac12 g_{jk}\dot z^j\dot z^k$, $\partial\mathcal{L}/\partial z^i=\tfrac12\partial_i g_{jk}\dot z^j\dot z^k$ and $\partial\mathcal{L}/\partial\dot z^i=g_{ij}\dot z^j$, so the Euler--Lagrange equation is
\[
\frac{d}{dt}\big(g_{ij}\dot z^j\big)-\tfrac12\partial_i g_{jk}\dot z^j\dot z^k
=g_{ij}\ddot z^j+\big(\partial_k g_{ij}-\tfrac12\partial_i g_{jk}\big)\dot z^j\dot z^k
=g_{i\ell}\big(\ddot z^\ell+\Gamma^\ell_{jk}\dot z^j\dot z^k\big)=0,
\]
i.e.\ $\tfrac{D}{dt}\dot z=0$, using $\Gamma^\ell_{jk}=\tfrac12 g^{\ell m}(\partial_j g_{mk}+\partial_k g_{mj}-\partial_m g_{jk})$. Since $\mathcal{L}$ has no explicit $t$-dependence, the energy integral $\dot z^i\partial\mathcal{L}/\partial\dot z^i-\mathcal{L}=\tfrac12 g_{ij}\dot z^i\dot z^j$ is conserved, so $\|\dot z\|_g$ is constant along extremals. Finally, by Cauchy--Schwarz,
\[
L(\gamma)=\int_{t_0}^{t_1}\|\dot z\|_g\,dt\le\Big(\int_{t_0}^{t_1}1\,dt\Big)^{1/2}\Big(\int_{t_0}^{t_1}\|\dot z\|_g^2\,dt\Big)^{1/2}=\sqrt{2(t_1-t_0)E[\gamma]},
\]
with equality if and only if $\|\dot z\|_g$ is constant. Thus a constant-speed energy-minimizer also minimizes length, and conversely a minimizing geodesic parameterized at constant speed minimizes $E$; the two variational problems share the same solution curves.
\end{proof}

\begin{theorem}[T6: small-noise concentration in regions of small potential gradient]\label{thm:T6}
Let $U\subset M_S$ satisfy $\|\nabla_{\!g}V\|_g\le\varepsilon$ on $U$. For fixed endpoints joined by a unique minimizing geodesic in $U$, the endpoint-conditioned law approaches concentration near that geodesic as $T\to0$ and $\varepsilon\to0$, subject to the bounded-path and correction-term conditions stated below.
\end{theorem}

\begin{proof}
Expanding the leading term of \eqref{eq:omfunctional} with $b=-\nabla_{\!g}V$ gives
\[
T\mathcal{A}_{\mathrm{OM}}
=\tfrac12 E[\gamma]
+\tfrac12\int_{t_0}^{t_1}\langle\dot z,\nabla_{\!g}V\rangle_g\,dt
+\tfrac14\int_{t_0}^{t_1}\|\nabla_{\!g}V\|_g^2\,dt
+T\mathcal{R}[z].
\]
For paths of uniformly bounded length in $U$, the absolute values of the second and third terms are bounded by $\tfrac12\varepsilon L(\gamma)$ and $\tfrac14\varepsilon^2(t_1-t_0)$, respectively. If $\mathcal{R}$ remains bounded on the path class, the scaled rate functional therefore converges to $\tfrac12E[\gamma]$ as $(T,\varepsilon)\to(0,0)$. By Lemma~\ref{lem:geodesic}, its unique fixed-endpoint minimizer is the assumed unique constant-speed minimizing geodesic. The endpoint-conditioned large-deviation principle of Theorem~\ref{thm:T5} \cite{freidlin2012,demboZeitouni2010}, together with stability of a unique rate minimizer under the stated bounds, then implies concentration in every neighborhood of that geodesic in the joint small-noise and small-gradient limit.
\end{proof}

\begin{corollary}[C6.1: conditional agreement of A3 and A4]\label{cor:C61}
Under the stated local-constancy and small-noise conditions, the planning minimizer of Axiom~\ref{ax:A4} and the most probable path of Axiom~\ref{ax:A3} select the same limiting geodesic.
\end{corollary}

\begin{proof}
On $U$ with $u_i$ and $V$ locally constant, the objective of \eqref{eq:A4} reduces (up to constants) to $\int T_S\,dt=E[\gamma]$, whose minimizer is the geodesic (Lemma~\ref{lem:geodesic}); by Theorem~\ref{thm:T6} this is also the $T\to0$ most probable path of Axiom~\ref{ax:A3}. Hence the two minimizers coincide.
\end{proof}

\subsubsection{Planning and learning}

These results concern Hypothesis~I and the planning component of Hypothesis~II.

\begin{theorem}[T7: first-order conditions and existence for planning]\label{thm:T7}
Under Axiom~\ref{ax:A4}, Definitions~\ref{def:energy},~\ref{def:utility} and Assumption~\ref{as:R3}, and a normality (constraint-qualification) condition, there exist multipliers $\nu_i\ge0$ such that a minimizer of \eqref{eq:A4} satisfies
\begin{equation}
\delta\!\int_{t_0}^{t_1}\Big(\sum_i\nu_i u_i-E_S\Big)dt=0,\qquad \nu_i\Big(\int u_i\,dt-c_i(t_1-t_0)\Big)=0,
\label{eq:T7kkt}
\end{equation}
equivalently the semantic action functional
\begin{equation}
J[z]=\int_{t_0}^{t_1}\big(U_S(z)-\lambda\,E_S(z,\dot z)\big)\,dt
\label{eq:action2}
\end{equation}
is obtained as a necessary condition with $\lambda$ a multiplier; and a minimizer exists.
\end{theorem}

\begin{proof}
The constraints $\int u_i\,dt\ge c_i(t_1-t_0)$ are isoperimetric. Under the normality condition (the active constraint gradients are linearly independent at the minimizer), the Lagrange multiplier theorem for isoperimetric problems \cite{dacorogna2008,lanczos1986} yields multipliers $\nu_i\ge0$ with the stationarity and complementary-slackness conditions \eqref{eq:T7kkt}; collecting $\sum_i\nu_i u_i=:U_S$ and dividing through gives \eqref{eq:action2} with $\lambda$ the (single active) multiplier, so $\lambda$ and the utility weights are multipliers rather than free parameters. Existence: the Lagrangian $E_S=\tfrac12 g_{ij}\dot z^i\dot z^j+V$ is, by Assumption~\ref{as:R3}, continuous in $z$, and convex and superlinear (coercive) in $\dot z$ because $g$ is uniformly positive definite; the integral functional is therefore sequentially weakly lower semicontinuous and coercive on the class of absolutely continuous curves with fixed endpoints, so by Tonelli's direct method \cite{dacorogna2008} a minimizer exists.
\end{proof}

\begin{corollary}[C7.1: interface degeneracy]\label{cor:C71}
If $\mathcal{A}=\varnothing$ then $u_{ctrl}\equiv0$ and $\nu_{ctrl}=0$, which follows from complementary slackness.
\end{corollary}

\begin{proof}
By Definition~\ref{def:utility}, $\mathcal{A}=\varnothing$ gives $u_{ctrl}\equiv0$, so $\int u_{ctrl}\,dt=0$. With the admissible budget $c_{ctrl}=0$ the control constraint is inactive, and \eqref{eq:T7kkt} forces $\nu_{ctrl}=0$.
\end{proof}

\begin{corollary}[C7.2: stationary condition for deliberation depth]\label{cor:C72}
An interior stationary depth $k$ satisfies $\partial_k\langle u_{reason}\rangle=\kappa'\,\partial_k C$, equating marginal reasoning utility with marginal computational cost.
\end{corollary}

\begin{proof}
On the learning scale, $k$ is selected by optimizing reasoning utility subject to computational cost: maximize $\langle u_{reason}\rangle(k)-\kappa'C(k)$ (the depth-dependent part of \eqref{eq:A1}). Stationarity in $k$ gives $\partial_k\langle u_{reason}\rangle-\kappa'\partial_k C=0$.
\end{proof}

\begin{theorem}[T8: learning monotonicity]\label{thm:T8}
Under Axiom~\ref{ax:A1} and Definitions~\ref{def:free}--\ref{def:complexity}, the projected descent flow makes $\Phi(\theta):=F_S[q_\theta]+\kappa C(\theta)$ non-increasing along learning time $s$. If, in addition, the entropy and complexity terms remain fixed on the active constraint set, then a decrease in $\Phi$ corresponds to a decrease in $\langle V\rangle_{q_s}$.
\end{theorem}

\begin{proof}
Axiom~\ref{ax:A1} is the (projected) gradient flow $\dot\theta=-P(\nabla\Phi)$, where $P$ is the projection onto the tangent cone of the feasible set $\{U_i\ge c_i\}$. Then
\[
\frac{d\Phi}{ds}=\nabla\Phi\cdot\dot\theta=-\nabla\Phi\cdot P(\nabla\Phi)\le0,
\]
since a Euclidean (or Riemannian, in the natural-gradient metric) projection satisfies $v\cdot P(v)=\|P(v)\|^2\ge0$. Thus $\Phi$ is non-increasing. Because it is bounded below (Assumption~\ref{as:R3}: $V$ bounded below, $C\ge0$, $S$ bounded on the constrained family), its value converges to a finite limit. Convergence of the parameter trajectory to a Karush--Kuhn--Tucker point would additionally require the standard compactness and regularity conditions for projected gradient flows and is not asserted by this calculation alone. When an active constraint satisfies $U_i=c_i$, and when the entropy and complexity terms are held fixed, a decrease in $\Phi$ corresponds to a decrease in $\langle V\rangle_{q_s}$. In the distributional form, the Jordan--Kinderlehrer--Otto minimizing-movement scheme $q_{k+1}=\arg\min_q F_S[q]+\tfrac{1}{2h}W_2^2(q,q_k)$ satisfies $F_S[q_{k+1}]\le F_S[q_k]$ by construction and converges to the Wasserstein gradient flow of $F_S$ \cite{jko1998,ambrosio2008}.
\end{proof}

\subsubsection{Reductions and multimodal alignment}

The following results examine objective-level correspondences and cross-modal alignment.

\begin{table}[htbp]
\caption{Objective-level correspondences between SLEP and selected existing frameworks.}\label{tab:reductions}
\begin{tabular}{@{}p{2.6cm}p{3.4cm}p{5.2cm}@{}}
\toprule
Theory & Instantiation & Reduction \\
\midrule
Information Bottleneck (IB) & $U_S\mapsto I(Z;Y)$, energy $\mapsto I(Z;X)$ & $J=I(Z;Y)-\lambda I(Z;X)$; learning scale \\
Variational autoencoder (VAE) / evidence lower bound (ELBO) & $U_S\mapsto\mathbb{E}_q[\log p_\psi]$, energy $\mapsto\mathrm{KL}(q\|p)$ & $\lambda=1\Rightarrow J=\mathrm{ELBO}$ \\
Predictive coding / Free Energy Principle (FEP) & $E\mapsto-\log p(x,z)$, $T=1$ & $F_S[q]=\mathrm{KL}(q\|p(z\mid x))-\log p(x)$; requires the distributional functional $F_S[q]$ \\
Maximum-entropy reinforcement learning (MaxEnt RL) & energy $\mapsto-Q_{\mathrm{soft}}$, $T\mapsto\alpha$ & $\pi^*\propto e^{Q_{\mathrm{soft}}/\alpha}$ \\
Natural gradient & flow $\mapsto$ steepest descent in $g$ & $\dot\theta=-g^{-1}\nabla L$ \\
World models & $V\mapsto$ latent predictive negative log-likelihood (NLL) & latent-dynamics learning $\approx\min\int V\,dt$ \\
\bottomrule
\end{tabular}
\end{table}

\begin{theorem}[T9: objective-level correspondences with existing frameworks]\label{thm:T9}
Under the substitutions in Table~\ref{tab:reductions}, selected objective functionals can be written in forms corresponding to components of the present framework. These correspondences do not imply equivalence of the full theories.
\end{theorem}

\begin{proof}
The substitutions are evaluated as follows.
(a) \emph{IB}: with $U_S\mapsto I(Z;Y)$ and energy $\mapsto I(Z;X)$, \eqref{eq:action2} reads $J=I(Z;Y)-\lambda I(Z;X)$, the IB Lagrangian \cite{tishby1999,alemi2017}; both terms are encoder functionals, so the correspondence applies on the learning scale.
(b) \emph{VAE}: with $U_S\mapsto\mathbb{E}_q[\log p_\psi(x\mid z)]$ and energy $\mapsto\mathrm{KL}(q\|p)$, at $\lambda=1$ we get $J=\mathbb{E}_q[\log p_\psi]-\mathrm{KL}(q\|p)=\mathrm{ELBO}$ \cite{kingma2014}.
(c) \emph{FEP}: with $E\mapsto-\log p(x,z)$ and $T=1$, Definition~\ref{def:free} gives $F_S[q]=\langle-\log p(x,z)\rangle_q-S[q]=\mathbb{E}_q[\log q(z)-\log p(x,z)]=\mathrm{KL}\big(q(z)\|p(z\mid x)\big)-\log p(x)$, Friston's variational free energy \cite{friston2010,friston2017}; this identity uses the distribution-level functional $F_S$ defined in Definition~\ref{def:free}.
(d) \emph{MaxEnt RL}: with energy $\mapsto-Q_{\mathrm{soft}}$ and $T\mapsto\alpha$, Theorem~\ref{thm:T3} gives the optimal policy $\pi^*\propto e^{Q_{\mathrm{soft}}/\alpha}$ \cite{levine2018,haarnoja2018}, which has the same exponential-family form in the decision setting.
(e) \emph{Natural gradient}: steepest descent in the geometry of Axiom~\ref{ax:A2} is $\dot\theta=-g^{-1}\nabla L$ \cite{amari1998}.
(f) \emph{World models}: in the cited world-model formulations, latent-dynamics learning includes minimization of latent predictive NLL, a Monte Carlo estimate of $\int V\,dt$ \cite{ha2018,hafner2025}.
\end{proof}

\begin{theorem}[T10: multimodal alignment]\label{thm:T10}
For multiple encoders sharing $(M_S,g,\psi)$ (Setting~\ref{set:S1}), if the communication/alignment constraint is active ($\nu_{comm}>0$), then a solution of Axiom~\ref{ax:A1} satisfies a cross-modal distortion bound $\mathbb{E}\,d_g\big(z^{(m)},z^{(m')}\big)^2\le D$. For a family of convex constrained problems indexed by the distortion budget $D$, the feasible cross-modal distortion bound is non-increasing as $D$ is reduced. Any monotone relation between the multiplier $\nu_{comm}$ and $D$ requires the usual sensitivity and regularity conditions.
\end{theorem}

\begin{proof}
Take $u_{comm}$ with receiver equal to another modality's encoder, so that the alignment constraint reads $\mathbb{E}\,d_g(z^{(m)},z^{(m')})^2\le D$ for a distortion budget $D$ (Definition~\ref{def:utility}). When active, complementary slackness (Theorem~\ref{thm:T7}, applied to \eqref{eq:A1}) gives equality $\mathbb{E}\,d_g^2=D$, yielding the stated bound. Across a family of otherwise identical constrained problems, reducing $D$ imposes a smaller admissible distortion. Complementary slackness alone does not establish a monotone relation between $\nu_{comm}$ and $D$; such a relation requires the standard sensitivity and regularity conditions for the parameterized convex problem.
\end{proof}

\emph{Remark (curvature).} The proposed relation between curvature and semantic complexity is treated as a heuristic interpretation of the geodesic-deviation equation, not as a theorem.

\subsubsection{Relation of the formal results to the hypotheses}

Section~\ref{sec:hypotheses} presents SLEP in a conceptual order, whereas the results of this section are organized according to mathematical dependency. Accordingly, the relationship between the hypotheses and the formal results is many-to-many rather than one-to-one: individual results may bear on multiple hypotheses, while each hypothesis may be developed through several results. The theorems, lemma, and corollaries do not by themselves establish the empirical validity of Hypotheses~I--V; rather, they specify the conditional mathematical consequences that follow from the setting, regularity assumptions, definitions, and axioms of this chapter. The relation between curvature and semantic complexity, by contrast, remains heuristic rather than a theorem-level consequence, as stated after Theorem~\ref{thm:T10}. Table~\ref{tab:hyp} summarizes the principal correspondence; it records formal correspondences only, and empirical assessment is deferred to the falsifiable predictions of the following section (Predictions~P1--P5).

\begin{table}[htbp]
\caption{Correspondence between the hypotheses of Section~\ref{sec:hypotheses} and the formal results of this section.}\label{tab:hyp}
\begin{tabular}{@{}p{3.0cm}p{2.5cm}p{5.0cm}@{}}
\toprule
Hypothesis (Section~\ref{sec:hypotheses}) & Principal formal results & Formal relation \\
\midrule
\textbf{I}: Low-energy latent semantic states & T7, C7.2, T8 & Derives the semantic action under utility constraints, identifies the marginal reasoning-utility versus computational-cost stopping rule, and establishes non-increase of the learning objective under the relevant conditions. \\
\textbf{II}: Variational semantic trajectories & T5, T6, C6.1, T7 & Characterizes endpoint-conditioned inference through the Onsager--Machlup variational principle, recovers geodesic behaviour in locally flat regions, and relates goal-directed planning to a derived semantic action. \\
\textbf{III}: Semantic geometry & T1, Lemma~1, T6, T10 & Establishes the well-posed Fisher--Riemannian semantic manifold, its geodesic organization, and a cross-modal geodesic-distortion bound. \\
\textbf{IV}: Semantic thermodynamics & T2, C2.1, T3, T4, C4.1 & Derives the invariant measure and Gibbs law, proves the semantic free-energy variational principle, identifies the stationary distribution, and gives an operational estimator of semantic temperature. \\
\textbf{V}: Semantic least-energy as a unifying principle & T9--T10; T1--T8 collectively & Identifies objective-level correspondences with several established frameworks and applies related variational and geometric constructions across modalities. \\
\bottomrule
\end{tabular}
\end{table}

\subsection{Predictions}\label{sec:falsify}

To limit post hoc adjustment, the evaluation protocol is specified before model assessment. Each energy component is nondimensionalized against a reference model, $\hat E_i:=E_i/E_i^{(\mathrm{ref})}$ (reference values from the initialization or a designated baseline), defining $V_{\mathrm{ref}}$ of Definition~\ref{def:energy} and the ``normalized semantic energy''. The number of coefficients estimated from evaluation data is limited: $k_B$ is absorbed into $T$; $\lambda$ and the utility weights are multipliers fixed by the $c_i$ (Theorem~\ref{thm:T7}); $C$ is defined on the learning scale (Definition~\ref{def:complexity}); and $T$ is estimated from the slope of Corollary~\ref{cor:C41}. Utility and energy are evaluated using separately specified data (held-out transfer for $U$; in-distribution NLL for $V$). The weights are estimated on a development family and then held fixed for evaluation on a disjoint family, with failure criteria specified in advance. Each prediction specifies its time scale and external reference measure.

\begin{prediction}[P1: energy decreases with learning; compression after fitting]\label{pred:P1}
$\widehat{F}_S$ and $\langle\hat V\rangle$ are non-increasing along $s$; after utility (held-out task performance) plateaus, $\langle\hat V\rangle$ continues to decrease over the prespecified post-plateau interval. \emph{Chain}: Axiom~\ref{ax:A1}$\to$Theorem~\ref{thm:T8}. \emph{Scale}: $s$. \emph{Anchor}: utility and energy are evaluated using separately specified held-out and in-distribution measurements. \emph{Failure}: $\langle\hat V\rangle$ stalls or rises after the plateau. This prediction can be compared with competing accounts of two-phase learning \cite{shwartzziv2017,saxe2019}.
\end{prediction}

\begin{prediction}[P2: trajectories concentrate on low-action Onsager--Machlup paths]\label{pred:P2}
The $\mathcal{A}_{\mathrm{OM}}$ quantile of observed trajectories is lower than that of endpoint-matched reference paths according to a prespecified effect-size and uncertainty criterion; path deviation has a prespecified positive association with externally defined novelty. \emph{Chain}: Axiom~\ref{ax:A3}$\to$Theorem~\ref{thm:T5}. \emph{Scale}: $t$. \emph{Anchor}: novelty criterion (membership in a prespecified out-of-distribution benchmark, label shift) independent of $E_{pred}$. \emph{Failure}: deviation is uncorrelated with, or negatively correlated with, external novelty under the prespecified association criterion.
\end{prediction}

\begin{prediction}[P3: deliberation trajectories in regions of small potential gradient]\label{pred:P3}
For regions satisfying a prespecified threshold on $\|\nabla_g\hat V\|$, endpoint-conditioned deliberation trajectories have smaller geodesic deviation than matched trajectories in higher-gradient regions; the deviation is positively associated with $\|\nabla_g\hat V\|$. \emph{Chain}: Axiom~\ref{ax:A3}$\to$Theorem~\ref{thm:T5}$\to$Theorem~\ref{thm:T6}. \emph{Scale}: $t$ (internal $\Psi$ evolution). \emph{Anchor}: the potential-gradient threshold is estimated from the $\hat V$ field independently of trajectory fitting. \emph{Failure}: geodesic deviation in low-gradient regions is not smaller than in matched higher-gradient regions under the prespecified criterion.
\end{prediction}

\begin{prediction}[P4: Boltzmann organization as an affine law; measurable temperature]\label{pred:P4}
Across states $\hat I$ is affine in $\hat V$; the split-specific estimates of $\hat T$ agree within a prespecified tolerance or uncertainty interval. \emph{Chain}: Axiom~\ref{ax:A3}$\to$Theorem~\ref{thm:T4}$\to$Corollary~\ref{cor:C41}. \emph{Scale}: equilibrium statistics of $t$. \emph{Anchor}: $\hat I$ and $\hat V$ are obtained from separately specified density and decoder-based estimators. \emph{Failure}: a prespecified lack-of-fit criterion rejects the affine model, or between-split variation exceeds the stated tolerance.
\end{prediction}

\begin{prediction}[P5: alignment as a function of the cross-modal objective weight]\label{pred:P5}
Cross-modal geodesic distortion is predicted to decrease as the prespecified cross-modal objective weight increases. Relative to the intervention model, the no-objective control is expected to show a smaller reduction in distortion. \emph{Chain}: Setting~\ref{set:S1}$+$Axiom~\ref{ax:A1}$\to$Theorem~\ref{thm:T10}. \emph{Scale}: $s$. \emph{Anchor}: controlled comparison of training with and without a cross-modal objective. \emph{Failure}: failure of a prespecified monotonic-trend test, or a between-condition effect within the prespecified equivalence or minimum-effect threshold.
\end{prediction}

\subsection{Scope and limitations}

The present formulation is not tied to a single architecture, optimization algorithm, or computational substrate, although its assumptions must be checked for each implementation. The energy, utility, and geometric quantities are operational definitions, and alternative formulations may be considered for other settings. Two boundaries are explicit. First (Assumption~\ref{as:R4}), the potential is defined under environmental stationarity; non-stationary environments are not covered here and would require an extension, for example through stochastic optimal control. Second (Axiom~\ref{ax:A4}), the planning principle is restricted to goal-directed segments; Corollary~\ref{cor:C61} identifies conditional agreement with the inference dynamics only in the regime stated there. The objective-level correspondence with predictive coding (Theorem~\ref{thm:T9}(c)) covers state estimation but not action selection (expected free energy).

Finally, we summarize the epistemic status of the axioms. The axioms state the principal empirical commitments of the framework, although empirical tests also depend on the regularity assumptions and operational definitions. Their current status is discussed in terms of conditional mathematical compatibility, relationships to existing objectives, and falsifiable predictions. Theorems~\ref{thm:T2} and~\ref{thm:T4} identify the same stationary density under their stated assumptions, while Corollary~\ref{cor:C61} identifies a restricted regime in which the planning and inference formulations select the same path. These results provide conditional compatibility checks rather than a proof of global consistency. Second, Table~\ref{tab:reductions} records objective-level correspondences with several established frameworks. These correspondences motivate the selected formulation but do not establish the empirical validity or uniqueness of the axioms. Third, each prediction P1--P5 specifies a derivation chain, an external reference measure, and a failure criterion (Section~\ref{sec:falsify}). A failed prediction would reduce support for the relevant axiom conditional on the regularity assumptions, operational definitions, and measurement protocol. Alternative explanations, including assumption violations and estimation error, would need to be assessed separately. The empirical standing of the axioms should therefore be updated in light of tests of Predictions~P1--P5, subject to the assumptions and measurement procedures described above.

\backmatter

\section*{Declarations}

\bmhead{Funding}
Not applicable.

\bmhead{Conflict of interest}
Not applicable.

\bmhead{Data availability}
Not applicable.

\bmhead{Code availability}
Not applicable.

\bmhead{Author contribution}
J.Z. conceived the central idea during discussions on poetry and math with J.J.Z. while walking through the historic campuses of Trinity College and King's College, University of Cambridge, over the Christmas and New Year period. H.Y.Z. and J.Z. developed the main theoretical framework. J.J.Z. and H.N.H. contributed to the development of the theoretical framework. All authors contributed to the Discussion, and participated in the writing and revision of the manuscript.

\bibliography{sn-bibliography}

\end{document}